\documentclass[english]{llncs}
\usepackage{amsmath}
\usepackage{amsfonts}
\usepackage{amssymb}
\usepackage{hyperref}
\usepackage{stmaryrd}
\usepackage{float}
\usepackage{bussproofs}

\begin{document}
\author{David Cerna \inst{1} \and Alexander Leitsch \inst{2}}
\institute{Research Institute for Symbolic
Computation (RISC) \\ Johannes Kepler University, Linz, Austria \\ \href{mailto:dcerna@risc.uni-linz.ac.at}{dcerna@risc.uni-linz.ac.at}
\and
Logic and Theory Group\\ Technical University of Vienna \\  \href{mailto: leitsch@logic.at}{ leitsch@logic.at} }

\title{Analysis of Clause set Schema Aided by Automated Theorem Proving: A Case Study}
\subtitle{[Extended Paper]}
\maketitle
\pagestyle{plain}

\begin{abstract}
The {\em schematic CERES method} \cite{CERESS2} is a recently developed method of cut elimination for {\em proof schemata}, that is a sequence of proofs with a recursive construction. Proof schemata can be thought of as a way to circumvent adding an induction rule to the \textbf{LK}-calculus. In this work, we formalize a schematic version of the {\em infinitary pigeonhole principle}, which we call the Non-injectivity Assertion schema (NiA-schema), in the \textbf{LKS}-calculus \cite{CERESS2}, and analyse the clause set schema extracted from the NiA-schema using some of the structure provided by the schematic CERES method. To the best of our knowledge, this is the first application of the constructs built for proof analysis of proof schemata to a mathematical argument since its publication. We discuss the role of {\em Automated Theorem Proving} (ATP) in schematic proof analysis, as well as the shortcomings of the schematic CERES method concerning the formalization of the NiA-schema, namely, the expressive power of the {\em schematic resolution calculus}.  We conclude with a discussion concerning the usage of ATP in schematic proof analysis. 
\end{abstract}

\section{Introduction}\label{sec:Introduction}

In Gentzen's {\em Hauptsatz}~\cite{Gentzen1935}, a sequent calculus for first order logic was introduced, namely, the \textbf{LK}-calculus. He then went on to show that the {\em cut} inference rule is redundant and in doing so, was able to show consistency of the calculus. The method he developed for eliminating cuts from \textbf{LK}-derivations works by inductively reducing the cuts in a given \textbf{LK}-derivation to cuts which either have a reduced {\em formula complexity} and/or reduced {\em  rank}~\cite{prooftheory}. This method of cut elimination is known as {\em reductive cut elimination}. A useful result of cut elimination for the \textbf{LK}-calculus is that cut-free \textbf{LK}-derivations have the {\em subformula property},  i.e. every formula occurring in the derivation is a subformula of some formula in the end sequent. This property allows for the construction of {\em Herbrand sequents} and other objects which are essential in proof analysis. 

However, eliminating cuts from \textbf{LK}-derivations does have its disadvantages, mainly concerning the number of computations steps needed and the size of the final cut-free proof. As pointed out by George Boolos in ``Don't eliminate cut'' ~\cite{Dontelimcut}, sometimes the elimination of cut inference rules from an \textbf{LK}-proof can result in an non-elementary explosion in the size of the proof. Though using cut elimination, it is also possible to gain mathematical knowledge concerning the connection between different proofs of the same theorem. For example, Jean-Yves Girard's application of reductive cut elimination to a variation of F\"{u}rstenberg-Weiss' proof of Van der Waerden's theorem ~\cite{ProocomWaerdens1987} resulted in the {\em analytic} proof of Van der Waerden's theorem as found by Van der Waerden himself. From the work of Girard, it is apparent that interesting results can be derived from eliminating cuts in ``mathematical'' proofs. 

A more recently developed method of cut elimination, the CERES method ~\cite{CERES}, provides the theoretic framework to directly study the cut structure of \textbf{LK}-derivations, and in the process reduces the computational complexity of deriving a cut-free proof. The cut structure is transformed into a clause set allowing for clausal analysis of the resulting clause form. Methods of reducing clause set complexity, such as {\em subsumption} and {\em tautology elimination} can be applied to the characteristic clause set to reduce its complexity. It was shown by Baaz \& Leitsch in ``Methods of cut Elimination''~\cite{Baaz:2013:MC:2509679} that this method of cut elimination has a {\em non-elementary speed up} over reductive cut elimination.

In the same spirit of Girard's work, Baaz et al. ~\cite{Baaz:2008:CAF:1401273.1401552} applied the CERES method to a formalized mathematical proof. At the time of applying the method to F\"{u}rstenberg's proof of the infinitude of primes, the CERES method had been generalized to {\em higher-order logic} ~\cite{CERESHIGH} and an attempt was made to apply this generalized method to to the formal version of F\"{u}rstenberg's proof. However, the tremendous complexity of the higher-order clause set \footnote{The individual clauses of the clause set were very large, some containing over 12 literals, and contained both higher order and first order free variables. Interactive theorem provers could not handle these clause sets, nor could a human adequately parse the clause set.} suggested the use of an alternative method. Instead of formalizing the proof as a single higher-order proof,  formalize it as a sequence of first-order proofs enumerated by a single numeric parameter, of which indexes the number of primes assumed to exists. The resulting schema of clause sets was refuted by a resolution schema resulting in Euclid's argument for prime construction. The resulting specification was produced on the mathematical meta-level.  At that time no object-level construction of the refutation schema existed. 

A mathematical formalizations of  F\"{u}rstenberg's proof requires induction. In the higher-order formalization, induction is easily formalized as part of the formula language. However in first-order, an induction rule needs to be added to the \textbf{LK}-calculus. As it was shown in ~\cite{CERESS2}, Reductive cut elimination does not work in the presence of an induction rule in the \textbf{LK}-calculus. Also, other systems~\cite{Mcdowell97cut-eliminationfor} which provided cut elimination in the presence of an induction rule do so at the loss of some essential properties, for example the subformula property.   

In ``Cut-Elimination and Proof Schemata''~\cite{CERESS2}, a version of the \textbf{LK}-calculus was introduced (\textbf{LKS}-calculus) allowing for the formalization of sequences of proofs as a single object level construction, i.e. the {\em proof schema}, as well as a framework for performing cut elimination on proof schemata. Cut elimination performed within the framework of  ~\cite{CERESS2} results in cut-free proof schemata with the subformula property. Essentially, the concepts found in ~\cite{CERES} were generalized to handle recursively defined proofs. It was shown in ~\cite{CERESS2} that {\em schematic} characteristic clause sets are always unsatisfiable, but it is not known whether a given schematic characteristic clause set will have a refutation expressible within the language provided for the resolution refutation schema. This gap distinguishes the schematic version of the CERES method from the previously developed versions. 

In this work, we continue the tradition outlined above of providing a case study of an application of a ``new'' method of cut elimination to a mathematical proof. Though our example is relatively less grand than the previously chosen proof it gives an example of a particularly hard single parameter induction. We chose the {\em tape proof}, found in  ~\cite{TAPEPROOFNOEQ,tapeproofpaper,TapeproofEX2}, and generalize it by considering a codomain of size $n$ rather than of size two. A well known variation of our generalization has been heavily studied in literature under the guise of the {\em Pigeonhole Principle} (PHP). Our generalization will be referred to as the {\em Non-injectivity  Assertion} (NiA). Though such a proof seems straight forward to formalize within the \textbf{LKS}-calculus, without a change to the construction used in~\cite{tapeproofpaper}, there was a forced {\em eigenvariable} violation. 

After formalizing the NiA as a proof schema (the NiA-schema) we apply the schematic CERES method. In our attempt to construct an ACNF schema~\cite{CERESS2} we heavily use Automated Theorem Provers (ATP), specifically SPASS ~\cite{SpassProver}, to develop the understanding needed for construction of such a schema. SPASS was used over other theorem provers mainly due to familiarity. How theorem provers were used in our attempt to construct an ACNF schema will be an emphasis of this work. As an end result, we were able to ``mathematically'' express an ACNF schema of the NiA-schema to a great enough extent to produce instances of the ACNF  in the \textbf{LK}-calculus; in a similar way as in the F\"{u}rstenberg's proof analysis~\cite{Baaz:2008:CAF:1401273.1401552}. Though, in our case we have a refutation for every instance (only the first few where found in~\cite{Baaz:2008:CAF:1401273.1401552}) . It remains an open problem whether a more expressive language is needed to express the ACNF of the NiA-schema in the framework of \cite{CERESS2} . We conjecture that ATP will play an important role in resolving this question as well as in future proof analysis using the schematic CERES method. 

The paper is structured as follows: In Sec. \ref{sec:SCERES}, we introduce the \textbf{LKS}-calculus and the essential concepts from~\cite{CERESS2} concerning the schematic clause set analysis. In Sec. \ref{sec:MathNiA} \& \ref{sec:FormNiA}, we formalize the NiA-schema in the \textbf{LKS}-calculus. In Sec. \ref{sec:CCSSE}, we extract the characteristic clause set from the NiA-schema and perform {\em normalization} and tautology elimination. In Sec. \ref{sec:ATPANAL}, we analysis the extracted characteristic clause set with the aid of SPASS. In Sec. \ref{sec:refuteset}, we provide a (``mathematically defined'') ACNF schema of the extracted characteristic clause set.  In Sec. \ref{sec:Conclusion}, we conclude the paper and discuss future work.

\section{The \textbf{LKS}-calculus and Clause set Schema}\label{sec:SCERES}

In this section we introduce the \textbf{LKS}-calculus which will be used to formalize the NiA-schema, and the parts of the schematic CERES method concerned with characteristic clause set extraction. We refrain from introducing the resolution refutation calculus provided in~\cite{CERESS2} because it does not particularly concern the work of this paper. Though we provide a resolution refutation of the characteristic clause set of the NiA-schema, there is a good reason to believe the constructed resolution refutation is outside the expressive power of the current schematic resolution refutation calculus. More specifically, the provided resolution refutation grows as a function of the free parameter $n$ with respect to a constant change in depth, i.e. grows wider faster than it grows deep. For more detail concerning the schematic CERES method, see~\cite{CERESS2}. 

\subsection{Schematic language, proofs, and the \textbf{LKS}-calculus}
The \textbf{LKS}-calculus is based on the \textbf{LK}-calculus constructed by Gentzen~\cite{Gentzen1935}. When one grounds the {\em parameter} indexing an \textbf{LKS}-derivation, the result is an  \textbf{LK}-derivation~\cite{CERESS2}. The term language used is extended to accommodate  the schematic constructs of  \textbf{LKS}-derivations. We work in a two-sorted setting containing a {\em schematic sort} $\omega$ and an {\em individual sort} $\iota$. The schematic sort only contains numerals constructed from the constant $0:\omega$, a monadic function $s(\cdot):\omega \rightarrow \omega$ and a single free variable, the free parameter indexing \textbf{LKS}-derivations, of which we represent using $n$. 

The individual sort is constructed in a similar fashion to the standard first order language~\cite{prooftheory} with the addition of schematic functions. Thus,  $\iota$ contains countably many constant symbols, countably many {\em constant function symbols}, and  {\em defined function symbols}. The constant function symbols are part of the  standard first order language and the defined function symbols are used for schematic terms. Though, defined function symbols can also unroll to numerals and thus can be of type $\omega^n \to \omega$.  The $\iota$ sort also has {\em free} and {\em bound} variables and an additional concept, {\em extra variables}~\cite{CERESS2}. These are variables introduced during the unrolling of defined function ({\em predicate}) symbols. We do not use extra variables in the formalization of the NiA-schema, but they are essential for the refutation of the characteristic clause set. Also important are the {\em schematic variable symbols} which are variables of type $\omega \rightarrow \iota$. Essentially second order variables, though, when evaluated with a {\em ground term} from the $\omega$ sort we treat them as first order variables. Our terms are built inductively using constants and variables as a base.  

Formulae are constructed inductively using countably many {\em constant predicate symbols} (atomic formulae), logical operators $\vee$,$\wedge$,$\rightarrow$,$\neg$,$\forall$, and $\exists$, as well as  {\em defined predicate symbols} which are used to construct schematic formulae. In this work {\em iterated $\bigvee$} is the only defined predicate symbol used, and of which has the following term algebra:
\begin{equation}
\label{eq:one}
\varepsilon_{\vee}= \bigvee_{i=0}^{s(y)} P(i) \equiv \left\lbrace \begin{array}{c}
{\displaystyle \bigvee_{i=0}^{s(y)} P(i) \Rightarrow \bigvee_{i=0}^{y} P(i) \vee P(s(y)) }\\
{\displaystyle \bigvee_{i=0}^{0} P(i) \Rightarrow P(0)}
\end{array}\right. 
\end{equation}
 From the above described term and formulae language we can provide the inference rules of the \textbf{LKE}-calculus, essentially the \textbf{LK}-calculus~\cite{prooftheory} plus an equational theory $\varepsilon$ (in our case $\varepsilon_{\vee}$ Eq. \ref{eq:one}). This theory, concerning our particular usage, is a primitive recursive term algebra describing the structure of the defined function(predicate) symbols. The \textbf{LKE}-calculus is the base calculus for the \textbf{LKS}-calculus which also includes {\em proof links} which will be described shortly.
\begin{definition}[$\varepsilon$-inference rule]

\begin{prooftree}
\AxiomC{$S\left[ t\right] $}
\RightLabel{$(\varepsilon)$}
\UnaryInfC{$S\left[ t'\right] $}
\end{prooftree}
In the $\varepsilon$ inference rule, the term $t$ in the sequent $S$ is replaced by a term $t'$ such that, given the equational theory  $\varepsilon$,  $\varepsilon \models t = t'$.
\end{definition}

To extend the \textbf{LKE}-calculus with  proof links we need a countably infinite set of {\em proof symbols}  denoted by $\varphi, \psi,\varphi_{i}, \psi_{j} \ldots$. Let $S(\bar{x})$ by a  sequent with schematic variables $\bar{x}$, then by the sequent $S(\bar{t})$ we use to denote the sequent $S(\bar{x})$ where each of the variables in $\bar{x}$ is replaced by the terms in the vector $\bar{t}$ respectively, assuming that they have the appropriate type. Let $\varphi$ be a proof symbol and $S(\bar{x})$ a sequent, then the expression \AxiomC{$(\varphi(\bar{t}))$}
\dashedLine
\UnaryInfC{$S(\bar{t})$}
\DisplayProof
is called a {\em proof link} . For a variable $n:\omega$, proof links
such that the only arithmetic variable is $n$ are called {\em $n$-proof links} \index{k-proof Link}.

\begin{definition}[\textbf{LKE}-calculus ~\cite{CERESS2}]
The sequent calculus $\mathbf{LKS}$
consists of the rules of $\mathbf{LKE}$, where proof links may appear
at the leaves of a proof.
\end{definition}

\begin{definition}[Proof schemata ~\cite{CERESS2}]\label{def.proofschema}
\index{Proof Schemata}
  Let $\psi$ be a proof symbol and $S(n,\bar{x})$ be a sequent
  such that $n:\omega$. Then a {\em proof schema pair for $\psi$} is a pair of $\mathbf{LKS}$-proofs $(\pi,\nu(k))$ with end-sequents $S(0,\bar{x})$ and $S(k+1,\bar{x})$ respectively such that $\pi$ may not contain proof links and $\nu(k)$ may
  contain only proof links of the form \AxiomC{$(\psi(k,\bar{a}))$}
  \dashedLine
  \UnaryInfC{$S(k,\bar{a})$}
  \DisplayProof 
and we say that it is a proof link to $\psi$. We call $S(n,\bar{x})$ the end sequent of $\psi$ and assume an identification between the formula occurrences in the end sequents of $\pi$ and $\nu(k)$ so that we can speak of occurrences in the end sequent of $\psi$. Finally a proof schema $\Psi$ is a tuple of proof schema pairs for $\psi_1 , \cdots \psi_\alpha$ written as $\left\langle \psi_1 , \cdots \psi_\alpha \right\rangle$, such that the $\mathbf{LKS}$-proofs for $\psi_{\beta}$ may also contain $n$-proof links to $\psi_{\gamma}$ for $1\leq \beta < \gamma\leq \alpha$. We also say that the end sequent of $\psi_1$ is a the end sequent of $\Psi$. 
\end{definition}

We will not dive further into the structure of proof schemata and instead refer the reader to~\cite{CERESS2}. We now introduce the {\em characteristic clause set schema}.

\subsection{Characteristic Clause set Schema}
The construction of the characteristic clause set as described for the CERES method~\cite{CERES} required inductively following the formula occurrences of cut formula ancestors up the proof tree to the leaves. However, in the case of proof schemata, the concept of ancestors and formula occurrence is more complex. A formula occurrence might be an ancestor of a cut formula in one recursive call and in another it might not. Additional machinery is necessary to extract the characteristic clause term from proof schemata. A set $\Omega$ of formula occurrences from the end-sequent of an LKS-proof $\pi$ is called {\em a configuration for $\pi$}. A configuration $\Omega$ for $\pi$ is called relevant w.r.t. a proof schema $\Psi$ if $\pi$ is a proof in $\Psi$ and there is a $\gamma \in \mathbb{N}$ such that $\pi$ induces a subproof $\pi$ of $\Psi \downarrow \gamma$
such that the occurrences in $\Omega$ correspond to cut-ancestors below $\pi$~\cite{thesis2012Tsvetan}. Note that the set of relevant cut-configurations can be computed given a proof schema $\Psi$. To represent a proof symbol $\varphi$ and configuration $\Omega$ pairing in a clause set we assign them a {\em clause set symbol} $cl^{\varphi,\Omega}(a,\bar{x})$, where $a$ is an  arithmetic term. 

\begin{definition}[Characteristic clause term ~\cite{CERESS2}]\label{def:charterm}
\index{Characteristic Term}
Let $\pi$ be an $\mathbf{LKS}$-proof and $\Omega$ a configuration. In the following, by $\Gamma_{\Omega}$ , $\Delta_{\Omega}$ and $\Gamma_{C}$ , $\Delta_{C}$ we will denote multisets of formulas of $\Omega$- and $cut$-ancestors respectively. Let $r$ be an inference in $\pi$. We define the clause-set term $\Theta_r^{\pi,\Omega}$ inductively:
\begin{itemize}
\item if $r$ is an axiom of the form $\Gamma_{\Omega} ,\Gamma_C , \Gamma \vdash \Delta_{\Omega} ,\Delta_C , \Delta$, then \\ $\Theta_{r}^{\pi,\Omega} = \left\lbrace \Gamma_{\Omega} ,\Gamma_C  \vdash \Delta_{\Omega} ,\Delta_C \right\rbrace $
\item if $r$ is a proof link of the form
\AxiomC{$\psi(a,\bar{u})$}
\dashedLine
\UnaryInfC{$\Gamma_{\Omega} ,\Gamma_C , \Gamma \vdash \Delta_{\Omega} ,\Delta_C , \Delta$}
\DisplayProof
then define $\Omega'$ as the set of formula occurrences from $\Gamma_{\Omega} ,\Gamma_C  \vdash \Delta_{\Omega} ,\Delta_C$ and $\Theta_{r}^{\pi,\Omega} = cl^{\psi,\Omega}(a,\bar{u})$
\item if $r$ is a unary rule with immediate predecessor \index{Predecessor} $r'$ , then $\Theta_{r}^{\pi,\Omega} =  \Theta_{r'}^{\pi,\Omega}$

\item if $r$ is a binary rule with immediate predecessors $r_1 $, $r_2 $, then 
\begin{itemize}
\item if the auxiliary formulas of $r$ are $\Omega$- or $cut$-ancestors, then
$\Theta_{r}^{\pi,\Omega} = \Theta_{r_1}^{\pi,\Omega} \oplus \Theta_{r_2}^{\pi,\Omega}$
\item otherwise, $\Theta_{r}^{\pi,\Omega} = \Theta_{r_1}^{\pi,\Omega} \otimes \Theta_{r_2}^{\pi,\Omega}$
\end{itemize}
\end{itemize}
Finally, define $\Theta^{\pi,\Omega} = \Theta_{r_0}^{\pi,\Omega}$ where $r_0$ is the last inference in $\pi$ and $\Theta^{\pi} = \Theta^{\pi,\emptyset}$. We call $\Theta^{\pi}$ the characteristic term of $\pi$. 
\end{definition}
Clause terms evaluate to sets of clauses by $|\Theta| = \Theta$ for clause sets $\Theta$, $|\Theta_1 \oplus \Theta_2| = |\Theta_1| \cup |\Theta_2$, $|\Theta_1 \otimes \Theta_2| = \{C \circ D \mid C \in |\Theta_1|, D \in |\Theta_2|\}$.

The characteristic clause term is extracted for each proof symbol in a given proof schema $\Psi$, and together they make the characteristic clause set schema for $\Psi$, $CL(\Psi)$.

\section{``Mathematical'' proof of the NiA Statement}\label{sec:MathNiA}

In this section we provide a mathematical proof of the NiA statement (Thm. \ref{thm:finalpart}). The proof is very close in structure to the formal proof written in the \textbf{LKS}-calculus, which can be found in Sec. \ref{sec:FormNiA}. We skip the basic structure of the proof and outline the structure emphasising the cuts. We will refer to the interval $\left\lbrace 0, \cdots, n-1 \right\rbrace $ as $\mathbb{N}_{n}$. Let $rr_{f}(n)$ be the following sentence, for $n\geq 2$: There exists $p,q \in \mathbb{N}$ such that  $p < q$ and  $f(p) = f(q)$,  or for all $x \in \mathbb{N}$ there exists a $y \in \mathbb{N}$ such that $x\leq y$ and $f(y)\in \mathbb{N}_{n-1}$.

\begin{lemma}
\label{lem:Inducbase}
Let $f:\mathbb{N} \rightarrow \mathbb{N}_{n}$, where $n\in \mathbb{N}$, be total, then  $rr_{f}(n)$ or there exists $p,q \in \mathbb{N}$ such that  $p < q$ and  $f(p) = f(q)$.
\end{lemma}
\begin{proof}
We can split the codomain into $\mathbb{N}_{n-1}$ and $\left\lbrace n \right\rbrace$, or the codomain is $\left\lbrace 0 \right\rbrace$.
\end{proof}
\begin{lemma}
\label{lem:inDucstep}
Let $f$ be a function as defined in Lem. \ref{lem:Inducbase} and $2< m\leq n$, then if $rr_{f}(m)$ holds so does $rr_{f}(m-1)$. 
\end{lemma}
\begin{proof}
Apply the steps of Lem. \ref{lem:Inducbase} to the right side of the {\em or} in $rr_{f}(m)$.
\end{proof}
\begin{theorem}
\label {thm:finalpart}
Let $f$ be a function as defined in Lem. \ref{lem:Inducbase} , then there exists $i,j \in \mathbb{N}$ such that $i<j$ and $f(i) = f(j)$.
\end{theorem}
\begin{proof}
Chain together the implications of  Lem. \ref{lem:inDucstep} and derive $rr_{f}(2)$, the rest is trivial by   Lem. \ref{lem:Inducbase}.
\end{proof}
This proof makes clear that the number of cuts needed to prove the statement is parametrized by the size of the codomain of the function $f$. The formal proof of the next section outlines more of the basic assumptions being that they are needed for constructing the characteristic clause set. 

\section{NiA formalized in the \textbf{LKS}-calculus}\label{sec:FormNiA}

In this section we provide a formalization of the NiA-schema whose proof schema representation is $\left\langle (\omega(0),\omega(n+1)),(\psi(0),\psi(n+1)) \right\rangle$. Cut-ancestors will be marked with a $^*$ and $\Omega$-ancestors with $^{**}$. Numerals (terms of the $\omega$ sort) will be marked with $\overline{\cdot}$. We will make the following abbreviations: $ EQ_{f} \equiv \exists p \exists q( p < q \wedge   f(p)= f(q))$, $I(\overline{n}) \equiv \forall x \exists y ( x\leq y \wedge \bigvee_{i=\overline{0}}^{\overline{n}} f(y) = \overline{i})$, $I_s(\overline{n}) \equiv \forall x \exists y ( x\leq y \wedge f(y) = \overline{n})$ and $AX_{eq}(\overline{n}) \equiv f(\beta) = \overline{n}^{*} ,f(\alpha) = \overline{n}^{*} \vdash  f(\beta)= f(\alpha)$ (the parts of $AX_{eq}(\overline{n})$ marked as cut ancestors are always cut ancestors in the NiA-schema).
\begin{figure}[H]
\begin{tiny}
\begin{prooftree}
\AxiomC{$\begin{array}{c} \vdash \alpha \leq \alpha^{*} \end{array}$}
\AxiomC{$\begin{array}{c}  f(\alpha) = \overline{0}  \vdash   f(\alpha) = \overline{0}^{*}\end{array}$}
\RightLabel{$\wedge :r$}
\BinaryInfC{$\begin{array}{c}\vdots\\ \forall x  f(x) = \overline{0}  \vdash I(\overline{0})^{*} \end{array}$}
\AxiomC{$\begin{array}{c}  s(\beta) \leq \alpha^{*} \vdash \beta < \alpha\end{array}$}

\AxiomC{$\begin{array}{c}   AX_{eq}(\overline{0})\end{array}$}
\RightLabel{$\wedge:r$}

\BinaryInfC{$\begin{array}{c} \vdots\\ I(\overline{0})^{*}  \vdash EQ_{f}\end{array}$}
\RightLabel{$cut$}
\BinaryInfC{$\begin{array}{c}  \forall x  f(x) = 0 \vdash  EQ_{f}\end{array}$}
\end{prooftree}
\end{tiny}
\caption{Proof symbol $\omega(0)$}
\end{figure}

\begin{figure}[H]
\begin{tiny}
\begin{prooftree}
\AxiomC{$\begin{array}{c}\varphi(\overline{n+1})  \end{array}$}
\dottedLine
\UnaryInfC{$\begin{array}{c}  I(\overline{n+1})^{*} \vdash    EQ_{f} \end{array}$}
\AxiomC{$\begin{array}{c} \vdash  \alpha \leq \alpha^{*} \end{array}$}

\AxiomC{$\begin{array}{c}  \bigvee_{i=\overline{0}}^{\overline{n+1}} f(\alpha) = \overline{i}  \vdash  \bigvee_{i=\overline{0}}^{\overline{n+1}} f(\alpha) = \overline{i}^{*}\end{array}$}
\RightLabel{$\wedge :r$}
\BinaryInfC{$\begin{array}{c}\vdots \\ \forall x \bigvee_{i=\overline{0}}^{\overline{n+1}} f(x) = \overline{i} \vdash   I(\overline{n+1})^{*} \end{array}$}
\RightLabel{$cut$}
\BinaryInfC{$\begin{array}{c}  \forall x \bigvee_{i=\overline{0}}^{\overline{n+1}} f(x) = \overline{i} \vdash    EQ_{f} \end{array}$}
\end{prooftree}

\end{tiny}
\caption{Proof symbol $\omega(n+1)$}
\end{figure}

\begin{figure}[H]
\begin{tiny}
\begin{prooftree}
\AxiomC{$\begin{array}{c}  s(\beta) \leq \alpha^{*} \vdash  \beta < \alpha\end{array}$}
\AxiomC{$\begin{array}{c}   AX_{eq}(\overline{0})\end{array}$}
\RightLabel{$\wedge:r$}

\BinaryInfC{$\begin{array}{c} \vdots\\ I_s(\overline{0})^{*} \vdash    EQ_{f}\end{array}$}
\end{prooftree}
\end{tiny}
\caption{Proof symbol $\psi(0)$}
\end{figure}

\begin{figure}[H]
\begin{tiny}
\begin{prooftree}
\AxiomC{$\begin{array}{c}   max(\alpha,\beta)\leq \gamma^{**}  \vdash \\  \alpha \leq \gamma^{*}  \end{array}$} 
\AxiomC{$\begin{array}{c}    f(\gamma) = \overline{0}^{**}     \vdash \\ f(\gamma) = \overline{0}^{*} ,\end{array}$}
\UnaryInfC{$\begin{array}{c}  \vdots \end{array}$}
\AxiomC{$\begin{array}{c}    f(\gamma) = \overline{n+1}^{**}     \vdash \\ f(\gamma) = n+1^{*} ,\end{array}$}
\BinaryInfC{$\begin{array}{c}\vdots    \end{array}$}
\RightLabel{$\wedge :r$}
\BinaryInfC{$\begin{array}{c}\vdots    \end{array}$}
\AxiomC{$\begin{array}{c}    max(\alpha,\beta)\leq \gamma^{**}  \vdash \\  \beta \leq \gamma^{*}  \end{array}$}
\RightLabel{$\wedge :r$}
\BinaryInfC{$\begin{array}{c} I(\overline{n+1})^{**} \vdash     I(\overline{n})^{*},    I_{s}(\overline{n+1})^{*} \\ \vdots \end{array}$}
\end{prooftree}
\end{tiny}

\begin{tiny}
\begin{prooftree}
\AxiomC{$\begin{array}{c}\vdots \\ I(\overline{n+1})^{**} \vdash     I(\overline{n})^{*},    I_{s}(\overline{n+1})^{*}\end{array}$}
\AxiomC{$\varphi(\overline{n})$}
\dottedLine
\UnaryInfC{$\begin{array}{c}  I(\overline{n})^{*}\vdash     EQ_{f}\end{array}$}
\RightLabel{$cut$}
\BinaryInfC{$\begin{array}{c}   I(\overline{n+1})^{**} \vdash   EQ_{f}, I_{s}(\overline{n+1})^{*}\\\vdots\end{array}$}
\end{prooftree}
\end{tiny}

\begin{tiny}
\begin{prooftree}
\AxiomC{$\begin{array}{c}\vdots \\  I(\overline{n+1})^{**} \vdash   EQ_{f}, I_{s}(\overline{n+1})^{*}\end{array}$}
\AxiomC{$\begin{array}{c}  s(\beta) \leq \alpha^{*} \vdash  \beta < \alpha\end{array}$}
\AxiomC{$\begin{array}{c}   AX_{eq}(\overline{n+1})\end{array}$}
\RightLabel{$\wedge:r$}

\BinaryInfC{$\begin{array}{c} \vdots\\I_{s}(\overline{n+1})^{*}  \vdash EQ_{f} \end{array}$}
\RightLabel{$cut$}
\BinaryInfC{$\begin{array}{c}   I(\overline{n+1})^{**} \vdash    EQ_{f} , EQ_{f} \end{array}$}
\RightLabel{$c:r$}
\UnaryInfC{$\begin{array}{c} I(\overline{n+1})^{**} \vdash    EQ_{f}  \end{array}$}
\end{prooftree}
\end{tiny}
\caption{Proof symbol $\psi(n+1)$}
\end{figure}

\section{Characteristic Clause set Schema Extraction }\label{sec:CCSSE}
The outline of the formal proof provided above highlights the inference rules which directly influence the characteristic clause set schema of the NiA-schema. Also to note are the configurations of the NiA-schema which are relevant, namely, the empty configuration $\emptyset$ and a schema of configurations $\Omega(\overline{n}) \equiv \forall x \exists y ( x\leq y \wedge \bigvee_{i=\overline{0}}^{\overline{n}} f(y) = \overline{i})$. Thus, we have the following:

\begin{subequations}
\label{seq:charclaset}
\begin{equation}
CL_{NiA}(0)\equiv \Theta^{\omega,\emptyset}(0)\equiv \left\lbrace \left( cl^{\psi,\Omega(\overline{0})}(\overline{0})\oplus \vdash \alpha\leq \alpha \right)\oplus \vdash f(\alpha)=\overline{0}  \right\rbrace 
\end{equation}
\begin{equation}
 cl^{\psi,\Omega(\overline{0})}(\overline{0}) \equiv\Theta^{\psi,\Omega(\overline{0})}(0) \equiv\left\lbrace  s(\beta)\leq   \alpha \vdash \otimes f(\alpha)=\overline{0},  f(\beta)=\overline{0}\vdash \right\rbrace 
\end{equation}
\begin{equation}
{\scriptstyle CL_{NiA}(\overline{n+1})\equiv \Theta^{\omega,\emptyset}(\overline{n+1})\equiv \left\lbrace \left( cl^{\psi,\Omega(\overline{n+1})}(\overline{n+1})\oplus \vdash \alpha\leq \alpha \right)\oplus \vdash \bigvee_{i=\overline{0}}^{\overline{n+1}} f(\alpha)=\overline{i}  \right\rbrace 
}\end{equation}
\begin{equation}
\begin{array}{c} {\scriptstyle cl^{\psi,\Omega(\overline{n+1})}(\overline{n+1})\equiv \Theta^{\psi,\Omega(\overline{n+1})}(\overline{n+1})\equiv \left\lbrace \left( cl^{\psi,\Omega(\overline{n})}(\overline{n}) \oplus \left( s(\beta)\leq \alpha \vdash \otimes f(\alpha)= \overline{n+1}, f(\beta)=\overline{n+1}\vdash \right) \right) \right. } \\ {\scriptstyle \left. \oplus \left( max(\alpha,\beta)\leq \gamma \vdash \alpha \leq \gamma \right)  \oplus  \left(max(\alpha,\beta)\leq \gamma \vdash    \beta \leq \gamma \right)  \right\rbrace } \end{array}
\end{equation}
\end{subequations}

In the characteristic clause set schema  $CL_{NiA}(\overline{n+1})$ presented in Eq.\ref{seq:charclaset} tautologies are already eliminated. {\em Evaluation} of $CL_{NiA}(\overline{n+1})$ yields the following clause set $C(n)$:
\[
\begin{array}{l}
(C1)\ \vdash \alpha  \leq  \alpha,\ (C2)\  max(\alpha, \beta) \leq \gamma \vdash \alpha \leq \gamma,\ (C3)\  max(\alpha, \beta)  \leq \gamma \vdash \beta \leq \gamma \\ 
(C4_{0})\  f(\beta)  = \overline{0} ,  f(\alpha)  = \overline{0}  ,   s(\beta) \leq \alpha  \vdash \\ 
\ldots  \ldots\\ 
\ldots \ldots  \\ 
(C4_{n})\ f(\beta)  = \overline{n}  ,  f(\alpha)  = \overline{n} ,   s(\beta)  \leq \alpha \vdash  \\ 
(C5)\ \vdash f(\alpha) = \overline{0} , \cdots , f(\alpha) = \overline{n} 
\end{array}
\]

\section{Clausal Analysis Aided by ATP}\label{sec:ATPANAL}

The result of characteristic clause set extraction for proof schemata is a sequence of clause sets representing the cut structure (See Sec. \ref{sec:CCSSE}), rather than a single clause set representing the cut structure. Thus, unlike applications of the first-order CERES method to formal proofs~\cite{tapeproofpaper}, where a theorem prover is used exclusively to find a refutation, we can only rely on theorem provers for suggestions. Essentially, we need the theorem provers to help with the construction of two elements of the schematic resolution refutation: the induction invariants and the term language. 

For this clause set analysis, we exclusively used SPASS~\cite{SpassProver} in the ``out of the box mode''. We did not see a point to working with the configurations of SPASS being that for sufficiently small instances of $C(n)$ it found a refutation, and our goal was not to find an elegant proof using the theorem prover, but rather a refutation with the aid of the theorem prover; the ``out of the box mode'' was enough for this goal\footnote{Also, using ``out of the box mode'' allows for ease of reproducibility of our results when using the same version of SPASS.}. Though as a side note, refutations found by SPASS were not the smallest, the resolution refutation that SPASS gave as output for $C(4)$\footnote{See Sec. \ref{sec:spassreffour}} used $(C5)$ in the refutation tree 1806 times. The resolution refutation we provide used $(C5)$ only 65 times. Though, it is not that our final refutation is wildly different, SPASS ended up deriving clauses using derived clauses which could easily  be derived from the initial clause set. 

An essential feature we were looking for in the refutations found by SPASS were sequences of clauses which mimic the stepcase construction of the induction axiom, i.e. $\forall x (\varphi(x) \rightarrow \varphi(x+1))$. An example of such a sequence from the refutation of $C(4)$, of which will be the basis of Thm. \ref{thm:refofC}, is as follows:
\begin{figure}
\begin{center}
\begin{tabular}{l}
$1[0:Inp] \ \|\|  \ \Rightarrow \ eq(f(U),3) \ , \ eq(f(U),2) \ , \ eq(f(U),1) \ , \ eq(f(U),0)*$\\
$2795[0:MRR:1.3,2764.0] \ \|\|  \ \Rightarrow \ eq(f(U),3) \ , \ eq(f(U),2) \ , \ eq(f(U),1)$\\
$3015[0:MRR:2795.2,2984.0] \ \|\|  \ \Rightarrow \ eq(f(U),3) \ , \ eq(f(U),2)$\\
$3096[0:MRR:3015.1,3065.0] \ \|\|  \ \Rightarrow \ eq(f(U),3)$
\end{tabular}
\end{center}
\caption{Recursive sequence found in the refutation of $C(4)$.}
\label{fig:seqone}
\end{figure}

Essentially, if we where to interpret the initial clause as defining a function (a function whose domain is the natural numbers and whose codomain\index{Codomain} is the set $\left[0,n \right]$) we see that at first we assume the function has a codomain of size $n$, and than we derive that it cannot have a codomain of size $n$, but rather of size $n-1$, and so on, until we derive that its codomain is empty, contradicting the original assumption, that is that the codomain is non-empty (i.e. clause $(C5)$). This pattern can be found in other instances of the refutation of $C(n)$.

This sequences seems to be an essential part, even the only part, needed to define a recursive refutation of $C(n)$, though if and only if, $C(n)$ is refutable with a {\em total induction}, of which such a refutation has not been found and is unlikely to exists. Something which is not completely apparent in SPASS refutation for $C(n)$, $n<4$, is the gap (in numbering) between clause $1$ and clause $2795$ in Fig. \ref{fig:seqone}. To derive clause $2795$ for clause $1$ in one step we need to first derive the following clause:
    $$2764[0:MRR:2714.1,2749.1] \ \|\| \ eq(f(U),0)* \ \Rightarrow $$
of which deriving is almost as difficult as deriving the sequence of Fig. \ref{fig:seqone}. Essentially to derive clause $2764$, the SPASS refutation eludes to the  need of an inner recursion bounded by the outer recursion. Essentially, we start from a clause of the following form: \\
$\begin{array}{l} 2272[0:Res:955.3,159.1] \ \|\| \ eq(f(U),0)* \ , \ eq(f(V),1)* \ , \  eq(f(W),2)* \ , \ \\ eq(f(X),3)* \ \Rightarrow \end{array}$\\
stating that the codomain is empty and derive that this implies some element $k$ is not in the codomain. Clause $2272$ is essential for Lem. \ref{lem:l13} and is one of the clauses of Lem. \ref{lem:lbase}.

Up to this point we have an idea of the overall structure of the refutation, but so far, we have not discussed the term structure and unifiers used by SPASS. Essentially, how was the recursive max term construction of Def. \ref{def:maxterm} found? Looking at the following two derived clauses from $C(3)$ and $C(4)$ we see that the nesting of the $\max$ term grows with respect to the free parameter:\\
$20[0:Res:15.0,4.0] \ \|\|  \ \Rightarrow \ le(U,max(max(max(V,U),W),X))$\\
$54[0:Res:19.0,4.0] \ \|\|  \ \Rightarrow \ le(U,max(max(max(max(V,U),W),X),Y))$\\
However, in clause $20$ and $54$ the associativity is the opposite of Def. \ref{def:maxterm}. We found that the refutation of Sec. \ref{sec:refuteset} is easier when we switch the association of the max term construction. Also, both clause $20$ and $54$  do not contain successor function ($s(\cdot )$) encapsulation of the variables while Def. \ref{def:maxterm} does. The $s(\cdot )$ terms were added because of the clauses $C4_i$. The literal $s(\alpha) \leq \beta$ enforces the addition of an $s(\cdot )$ term anyway during the unification. This can be see in Lem. \ref{lem:first} and  Cor. \ref{cor:first}, \ref{cor:second}, \& \ref{cor:third}. However, we have not been able to prove the necessity of these max function constructions, nor find a refutation without them. 

The result of all these observations was Lem. \ref{lem:lbase}. After proving that the Lem. \ref{lem:lbase} clause set is indeed derivable from $C(n)$ using resolution, we constructed it to see what the SPASS refutation looked like for $C(4)$. We abbreviate the term $max(max(max(s(x_{0}),s(x_1)),s(x_2)),s(x_3))$ by $m(\bar{x}_{4})$: 
\begin{figure}[H]
$1: eq(f(m(\bar{x}_{4})),2) \vee eq(f(m(\bar{x}_{4})),1) \vee eq(f(m(\bar{x}_{4})),0)$\\
$2:\neg eq(f(x_2),2) \vee eq(f(m(\bar{x}_{4})),1) \vee eq(f(m(\bar{x}_{4})),0)$\\
$3:\neg eq(f(x_1),1) \vee  eq(f(m(\bar{x}_{4})),2) \vee  eq(f(m(\bar{x}_{4})),0)$\\
$4:\neg eq(f(x_0),0) \vee  eq(f(m(\bar{x}_{4})),2) \vee  eq(f(m(\bar{x}_{4})),1)$\\
$5:\neg eq(f(x_2),2) \vee \neg eq(f(x_1),1) \vee eq(f(m(\bar{x}_{4})),0)$\\
$6: \neg eq(f(x_2),2) \vee \neg eq(f(x_0),0) \vee eq(f(m(\bar{x}_{4})),1)$\\
$7:\neg eq(f(x_1),1) \vee \neg eq(f(x_0),0)\vee eq(f(m(\bar{x}_{4})),2)$\\
$8:\neg eq(f(x_1),1) \vee \neg eq(f(x_0),0) \vee \neg eq(f(x_{2}),2)$
\caption{Clause set of Lem. \ref{lem:lbase} for $C(3)$.}
\end{figure}

Feeding this derived clause set to SPASS for several instances aided the construction of the well ordering of Def. \ref{def:importantordering} and the structure of the resolution refutation found in Lem. \ref{lem:l13}.

\section{Refutation of the NiA-schema's Characteristic Clause Set Schema} \label{sec:refuteset}
In this section we provide a refutation of  $C(n)$ for every value of $n$. We prove this result by first deriving a set of clauses which we will consider the least elements of a well ordering. Then we show how resolution can be applied to this least elements to derive clauses of the form $f(\alpha)= \overline{i} \vdash $ for $0\leq i \leq n$. The last step is simply to take the clause $(C5)$ from the clause set $C(n)$ and resolve it with each of the  $f(\alpha) = \overline{i}\vdash $ clauses.

\begin{definition}
\index{Iterated Max term}
\label{def:maxterm}
We define the primitive recursive term $m(k,\overline{x},t)$, where $\overline{x}$ is a schematic variable and $t$ a term, as follows: $\left\lbrace m(k+1,\overline{x},t) \Rightarrow  \right. $\\ $ \left. m(k,\overline{x},max(s(x_{k+1}),t))  \ ; \ m(0,t) \Rightarrow t \right\rbrace$

\end{definition}

\begin{definition}
\index{Resolution Rule}
\label{def:resstep}
We define the resolution rule $res(\sigma,P)$ where $\sigma$ is a unifier and $P$ is a predicate as follows:
\begin{prooftree}
\AxiomC{$\begin{array}{c}\Pi \vdash P^*, \Delta \end{array}$}
\AxiomC{$\begin{array}{c} \Pi' , P^{**}  \vdash \Delta'  \end{array}$}
\RightLabel{$res(\sigma,P)$} 
\BinaryInfC{$\begin{array}{c}\Pi\sigma , \Pi'\sigma \vdash \Delta\sigma , \Delta'\sigma\end{array}$}
\end{prooftree}
 The predicates $P^*$ and $P^{**}$ are defined such that $P^{**}\sigma = P^*\sigma = P$. Also, there are no occurrences of $P$ in $\Pi'\sigma$ and $P$ in $\Delta\sigma$.
\end{definition}
This version of the resolution rule is not complete for unsatisfiable clause sets, but  simplifies the outline of the refutation.

\begin{lemma}
\label{lem:first}
Given $0\leq k$ and $0 \leq n$, the clause $\vdash t \leq m(k,\overline{x},t)$ is derivable by resolution from $C(n)$.
\end{lemma}
\begin{proof}
Let us consider the case when $k=0$, the clause we would like to show derivability of is $\vdash t \leq m(0,t)$, which is equivalent to the clause $\vdash t \leq t$, an instance of (C1).
Assuming the lemma holds for all $m<k+1$, we show that the lemma holds for $k+1$. By the  induction hypothesis, the instance  $\vdash max(s(x_{k+1}),t') \leq m(k,\overline{x},max(s(x_{k+1}),t'))$ is  derivable. Thus, the following derivation proves that the clause $\vdash t' \leq m(k+1,\overline{x}_{k+1},t')$, where $t= max(s(x_{k+1}),t')$ for some term $t'$ is derivable:

\begin{prooftree}

\AxiomC{$\begin{array}{c}(IH) \\ \vdash P \end{array}$}
\AxiomC{$\begin{array}{c}  (C3) \\ max(\beta ,\delta ) \leq \gamma  \vdash  \delta \leq  \gamma  \end{array}$}
\RightLabel{$res(\sigma,P)$} 
\BinaryInfC{$\begin{array}{c}\vdash  t \leq m(k,\overline{x},max(s(x_{k+1}),t)) \end{array}$}
\RightLabel{$\varepsilon$}

\UnaryInfC{$\begin{array}{c}\vdash  t \leq m(k+1,\overline{x},t) \end{array}$}
\end{prooftree}
$$P = max(s(x_{k+1}),t) \leq  m(k,\overline{x},max(s(x_{k+1}),t))$$
$$\sigma =\left\lbrace \beta  \leftarrow s(x_{k+1}),  \gamma \leftarrow m(k,\overline{x},max(s(x_{k+1}),t)) ,  \delta \leftarrow t \right\rbrace $$
\\
$\square$
\end{proof}
See Sec. \ref{Appendix} for proofs of the following three corollaries.
\begin{corollary}
\label{cor:first}
Given $0\leq k,n$, the clause $ \vdash s(x_{k+1})\leq m(k,\overline{x},max(s(x_{k+1}),t))$ is derivable by resolution from $C(n)$.
\end{corollary}

\begin{corollary}
\label{cor:second}
Given $0\leq k$ and $0 \leq n$, the clause $  f(x_{k+1})= i, $\\ $f(m(k,\overline{x},max(s(x_{k+1}),t))) = i \vdash$ for $0\leq i \leq n$ is derivable by resolution from $C(n)$.
\end{corollary}

\begin{corollary}
\label{cor:third}
Given $0\leq k$ and $0 \leq n$, the clause $  f(x_{k+1})= i, f(m(k,\overline{x}_{k},s(x_{k+1}))) = i \vdash$ for $0\leq i \leq n$ is derivable by resolution from $C(n)$.
\end{corollary}

\begin{definition}
Given $0 \leq n$, $-1\leq k \leq j \leq n$,a variable $z$, and a bijective function $b: \mathbb{N}_{n} \rightarrow \mathbb{N}_{n}$ we define the following formulae: 
\begin{equation*}
c_{b}(k,j,z) \equiv \bigwedge_{i=0}^{k} f(x_{b(i)}) = b(i) \vdash \bigvee_{i=k+1}^{j} f(m(n,\overline{x},z)) = b(i).
\end{equation*}
The formulae $c_{b}(-1,-1,z) \equiv \ \vdash$, and $c_{b}(-1,n,z) \equiv \ \vdash \bigvee_{i=0}^{n} f(z) = i$ for all values of $n$ .
\end{definition}

\begin{lemma}
\label{lem:lbase}
Given $0 \leq n$, $-1\leq k \leq n$ and for all bijective functions $b: \mathbb{N}_{n} \rightarrow \mathbb{N}_{n}$. the formula $c_{b}(k,n,z)$ is derivable by resolution from C(n).
\end{lemma}
\begin{proof}
See Sec. \ref{sec:lbaseproof}. Greatest lower bounds of Def. \ref{def:importantordering}.
\end{proof}

\begin{definition}
Given $0 \leq n$, $0\leq k \leq j \leq n$, and a bijective function $b: \mathbb{N}_{n} \rightarrow \mathbb{N}_{n}$ we define the following formulae:

\begin{equation*}
c'_{b}(k,j) \equiv \bigwedge_{i=0}^{k} f(x_{i+1}) = b(i) \vdash \bigvee_{i=k+1}^{j} f(m(k,\overline{x}_{k},s(x_{k+1})) = b(i).
\end{equation*}
\end{definition}
\begin{lemma}
\label{lem:lbase2}
Given $0 \leq n$, $0\leq k \leq n$ and for all bijective functions $b: \mathbb{N}_{n} \rightarrow \mathbb{N}_{n}$. the formula $c'_{b}(k,n)$ is derivable by resolution from C(n).
\end{lemma}
\begin{proof}
See Sec. \ref{sec:lbase2proof}.
\end{proof}
\begin{definition}
\index{Clause Set Ordering}
\label{def:importantordering}
Given $0\leq n$ we define the ordering relation $\lessdot_{n}$ over $A_{n} = \left\lbrace (i,j) | i\leq j \right. $ $\left. \wedge 0 \leq i,j \leq n \wedge i,j \in \mathbb{N} \right\rbrace$
 s.t. for $(i,j),(l,k) \in A_n$, $(i,j) \lessdot_{n} (l,k)$ iff  $i,k,l \leq n$, $j<n$, $l\leq i$, $k\leq j$, and $i = l \leftrightarrow j \not = k$ and $j = k \leftrightarrow i \not = l$.
\end{definition}

\begin{lemma}
\index{Complete Well Ordering}
The ordering $\lessdot_{n}$ over $A_{n}$ for $0\leq n$ is a complete well ordering.
\end{lemma}
\begin{proof}
Every chain has a greatest lower bound\index{Greatest Lower Bound}, namely, one of the members of $A_{n}$, $(i,n)$  where $0\leq i \leq n$, and it is transitive, anti-reflexive, and anti-symmetric.
\end{proof}

The clauses proved derivable by Lem. \ref{lem:lbase2} can be paired with members of  $A_{n}$ as follows,  $c'_{b}(k,n)$ is paired with $(k,n)$. Thus, each $c'_{b}(k,n) $ is essentially the greatest lower bound of some chain in the ordering $\lessdot_{n}$ over $A_{n}$.

\begin{lemma}
\label{lem:l13}
Given  $0\leq k \leq j\leq n$, for all bijective functions $b: \mathbb{N}_{n} \rightarrow \mathbb{N}_{n}$ the clause $c'_{b}(k,j)$ is derivable from C(n).
\end{lemma}

\begin{proof}
We will prove this lemma by induction over $A_{n}$.  The base cases are the clauses $c'_{b}(k,n)$ from Lem. \ref{lem:lbase2}.  Now let us assume that the lemma holds for all clauses  $c'_{b}(k,i)$  pairs such that, $0\leq k \leq j <i \leq n$ and for all clauses $c'_{b}(w,j)$ such that $0\leq k< w \leq j\leq n$, then we want to show that the lemma holds for the clause $c'_{b}(k,j)$. We have not made any restrictions on the bijections used, we will need two different bijections to prove the theorem. The following derivation provides proof:
\begin{prooftree}
\AxiomC{$\begin{array}{c}(IH[k,j+1]) \\ \Pi_{b}(k), \vdash \Delta_{b}(k,j), P_{b}(j+1)\end{array}$}
\AxiomC{$\begin{array}{c}  (IH[k+1,k+1]) \\  \Pi_{b'}(k),  f(x_{b'(k+1)}) = b'(k+1) \vdash  \end{array}$}
\RightLabel{$res(\sigma,P)$} 
\BinaryInfC{$\begin{array}{c}   \Pi_{b}(k), \Pi_{b'}(k) \vdash \Delta_{b}(k,j) \end{array}$}
\RightLabel{$c:l$} 
\UnaryInfC{$\begin{array}{c}   \Pi_{b}(k)\vdash \Delta_{b}(k,j) \\ c'_{b}(k,j)\end{array}$}
\end{prooftree}
\begin{minipage}{.57\textwidth}
 $P_{b}(k+1) = f(m(k,\overline{x}_{k},s(x_{k+1}))) = b(k+1)$,
 \end{minipage}
\begin{minipage}{.5\textwidth}
$\Pi_{b}(k) \equiv \bigwedge_{i=0}^{k} f(x_{b(i)}) = b(i)$,
\end{minipage}
$$\Delta_{b}(k,j) \equiv \bigvee_{i=k+1}^{j} f(m(k,\overline{x}_{k},s(x_{k+1}))) = b(i), $$ 
$$\sigma =\left\lbrace x_{b'(k+1)}  \leftarrow m(k,\overline{x}_{k},s(x_{k+1}))\right\rbrace $$

We assume that $b'(k+1) = b(j+1)$ and that $b'(x)=b(x)$ for $0 \leq x \leq k$. 
\end{proof}

\begin{theorem}
\label{thm:refofC}
Given $n \geq 0$, $C(n)$ derives $\vdash$. 
\end{theorem}

\begin{proof}
By Lem. \ref{lem:l13}, The clauses $f(x)=  0 \vdash $ , $\cdots$ , $f(x)= n \vdash $ are derivable. Thus, we can prove the statement by induction on the instantiation of the clause set. When $n=0$, the clause (C5) is $\vdash f(x)=  0$ which resolves with $f(x)=  0 \vdash $ to derive $\vdash$. Assuming that for all $n'\leq n$ the theorem holds we now show that it holds for $n+1$. The clause (C5) from the clause set $C(n+1)$ is the clause (C5) from the clause set $C(n)$ with the addition of a positive instance of $\vdash f(\alpha)= (n+1)$. Thus, by the induction hypothesis we can derive the clause $\vdash f(\alpha)= (n+1)$. By Lem. \ref{lem:l13} we can derive $f(x)= (n+1) \vdash $, and thus, resolving the two derived clauses results in $\vdash$.

\end{proof}

\section{Conclusion} 

\label{sec:Conclusion}

At the end of the introduction, we outlined some essential points to be addressed in future work, i.e. finding a refutation which fits the framework of ~\cite{CERESS2} or showing that it is not possible and constructing a more expressive language.  Concerning the compression (see Sec. \ref{growthrate}), knowing the growth rate of the ACNF can help in the construction of a more expressive language for the refutations, and will be part of the future investigation. However, there is an interesting points which was not addressed, namely extraction of a {\em Herbrand system}. The extraction of  {\em Herbrand system} is the theoretical advantage this framework has over the previously investigated system~\cite{Mcdowell97cut-eliminationfor}\footnote{The schematic CERES method has the subformula property.} for cut elimination in the presence of induction, but without a refutation within the expressive power of the resolution calculus, the method of~\cite{CERESS2} cannot be used to extract a Herbrand system from our refutation. We plan to investigate the extraction of a Herbrand system for the NiA-schema given the current state of the proof analysis. Development of such a method can help find Herbrand systems in other cases when the ACNF-schema cannot be expressed in the calculus provided in~\cite{CERESS2}.

\bibliographystyle{plain}
\bibliography{references}

\section{Appendix}
\label{Appendix}

\subsection{Proof of Lem. \ref{cor:first}}
\begin{prooftree}
\AxiomC{$\begin{array}{c}(Lem. \ref{lem:first}) \\ \vdash P \end{array}$}
\AxiomC{$\begin{array}{c}  (C2) \\ max(\beta ,\delta ) \leq \gamma  \vdash  \beta \leq  \gamma  \end{array}$}
\RightLabel{$res(\sigma,P)$} 
\BinaryInfC{$\begin{array}{c}\vdash  s(x_{k+1}) \leq m(k,\overline{x},max(s(x_{k+1}),t)) \end{array}$}
\end{prooftree}
$$P = max(s(x_{k+1}),t) \leq  m(k,\overline{x},max(s(x_{k+1}),t))$$
$$\sigma =\left\lbrace \beta  \leftarrow s(x_{k+1}),  \gamma \leftarrow m(k,\overline{x},max(s(x_{k+1}),t)) ,  \delta \leftarrow t \right\rbrace $$
\\
$\square$

\subsection{Proof of Cor. \ref{cor:second}}
\begin{prooftree}
\AxiomC{$\begin{array}{c}(Cor. \ref{cor:first}) \\ \vdash P \end{array}$}
\AxiomC{$\begin{array}{c}  (C4_i) \\ f(\alpha)= i , f(\beta) = i , s(\alpha)\leq \beta  \vdash   \end{array}$}
\RightLabel{$res(\sigma,P)$} 
\BinaryInfC{$\begin{array}{c}   f(x_{k+1})= i, f(m(k,\overline{x}_{k},max(s(x_{k+1}),t))) = i \vdash \end{array}$}
\end{prooftree}
$$P = s(x_{k+1}) \leq  m(k,\overline{x}_{k},max(s(x_{k+1}),t))$$
$$\sigma =\left\lbrace \alpha  \leftarrow x_{k+1},  \beta \leftarrow m(k,\overline{x}_{k},max(s(x_{k+1}),t)) \right\rbrace $$
\\
$\square$

\subsection{Proof of Cor. \ref{cor:third}}
\begin{prooftree}
\AxiomC{$\begin{array}{c}(Lem. \ref{lem:first}) \\ \vdash P \end{array}$}
\AxiomC{$\begin{array}{c}  (C4_i) \\ f(\alpha)= i , f(\beta) = i , s(\alpha)\leq \beta  \vdash   \end{array}$}
\RightLabel{$res(\sigma,P)$} 
\BinaryInfC{$\begin{array}{c}   f(x_{k+1})= i, f(m(k,\overline{x},s(x_{k+1}))) = i \vdash \end{array}$}
\end{prooftree}
$$P = s(x_{k+1}) \leq  m(k,\overline{x}_{k},s(x_{k+1}))$$
$$\sigma =\left\lbrace \alpha  \leftarrow x_{k+1},  \beta \leftarrow m(k,\overline{x}_{k},s(x_{k+1}))) \right\rbrace $$
\\
$\square$
\subsection{Proof of Cor. \ref{lem:lbase}}
\label{sec:lbaseproof}
We prove this lemma by induction on $k$ and a case distinction on $n$.  When $n=0$ there are two possible values for $k$, $k=0$ or $k=-1$. When $k=-1$ the clause is an instance of (C5). When $k=0$ we have the following derivation:
\begin{prooftree}
\AxiomC{$\begin{array}{c}(C5) \\ c_{b}(-1,1,y) \end{array}$}
\AxiomC{$\begin{array}{c}  (Cor. \ref{cor:second}[i\leftarrow b(0), k\leftarrow 0 ]) \\ f(x_{1})= b(0), f(max(s(x_{1}),z)) = b(0) \vdash  \end{array}$}
\RightLabel{$res(\sigma,P)$} 
\BinaryInfC{$\begin{array}{c}   c_{b}(0,1,z)  \end{array}$}
\end{prooftree}
$$P = f(max(s(x_{1}),z)) = b(0)$$
$$\sigma =\left\lbrace y  \leftarrow max(s(x_{1}),z)\right\rbrace $$

By $(Cor. \ref{cor:second}[i\leftarrow b(0), k\leftarrow 0 ])$ we mean take the clause that is proven derivable by Cor. \ref{cor:second} and instantiate the free parameters of Cor. \ref{cor:second}, i.e. $i$ and $k$, with the given terms, i.e. $b(0)$ and $0$.   Remember that $b(0)$ can be either $0$ or $1$. We will use this syntax through out the dissertion. When $n>0$ and  $k=-1$ we again trivially have (C5). When $n>0$ and  $k=0$, the following derivation suffices:

\begin{prooftree}
\AxiomC{$\begin{array}{c}(C5) \\ c_{b}(-1,n,y) \end{array}$}
\AxiomC{$\begin{array}{c}  (Cor. \ref{cor:second}[i\leftarrow b(0), k\leftarrow 0 ]) \\ f(x_{1})= b(0), f(max(s(x_{1}),z)) = b(0) \vdash  \end{array}$}
\RightLabel{$res(\sigma,P)$} 
\BinaryInfC{$\begin{array}{c}   c_{b}(0,n,z)  \end{array}$}
\end{prooftree}
$$P = f(max(s(x_{1}),z)) = b(0)$$
$$\sigma =\left\lbrace y  \leftarrow max(s(x_{1}),z)\right\rbrace $$

The main difference between the case for $n=1$ and $n>1$ is the possible instantiations of the bijection at $0$. In the case of $n>1$, $b(0) = 0 \ \vee  \cdots \vee \ b(0) = n$.  Now we assume that for all $w< k+1 <n$ and $n>0$ the theorem holds, we proceed to show that the  theorem holds for $k+1$.  The following derivation will suffice:

\begin{prooftree}
\AxiomC{$\begin{array}{c}(IH) \\ c_{b}(k,n,y) \end{array}$}
\AxiomC{$\begin{array}{c}  (Cor. \ref{cor:second}[i\leftarrow b(k+1)]) \\  f(x_{k+1})= b(k+1), P \vdash  \end{array}$}
\RightLabel{$res(\sigma,P)$} 
\BinaryInfC{$\begin{array}{c}   c_{b}(k+1,n,z)  \end{array}$}
\end{prooftree}
$$P = f(m(k,\overline{x}_{k},max(s(x_{k+1}),t))) = b(k+1)$$
$$\sigma =\left\lbrace y  \leftarrow max(s(x_{k+1}),z)\right\rbrace $$
\\
$\square$

\subsection{Proof of Lem. \ref{lem:lbase2}}
\label{sec:lbase2proof}
We prove this lemma by induction on $k$ and a case distinction on $n$.  When $n=0$ it must be the case that $k=0$. When $k=0$ we have the following derivation :
\begin{prooftree}
\AxiomC{$\begin{array}{c}(C5) \\ c_{b}(-1,0,y) \end{array}$}
\AxiomC{$\begin{array}{c}  (Cor. \ref{cor:third}[i\leftarrow 0, k\leftarrow 0 ]) \\ f(x_{1})= 0, f(s(x_{1})) = 0 \vdash  \end{array}$}
\RightLabel{$res(\sigma,P)$} 
\BinaryInfC{$\begin{array}{c}   c'_{b}(0,0)  \end{array}$}
\end{prooftree}
$$P = f(s(x_{1})) = 0$$
$$\sigma =\left\lbrace y  \leftarrow s(x_{1})\right\rbrace $$

Remember that $b(0)$ can only be mapped to  $0$. When $n>0$ and  $k=0$, the following derivation suffices:

\begin{prooftree}
\AxiomC{$\begin{array}{c}(C5) \\ c_{b}(-1,n,y) \end{array}$}
\AxiomC{$\begin{array}{c}  (Cor. \ref{cor:third}[i\leftarrow b(0), k\leftarrow 0 ]) \\ f(x_{1})= b(0), f(s(x_{1})) = b(0) \vdash  \end{array}$}
\RightLabel{$res(\sigma,P)$} 
\BinaryInfC{$\begin{array}{c}   c'_{b}(0,n)  \end{array}$}
\end{prooftree}
$$P = f(s(x_{1})) = b(0)$$
$$\sigma =\left\lbrace y  \leftarrow s(x_{1})\right\rbrace $$

The main difference between the case for $n=0$ and $n>0$ is the possible instantiations of the bijection at $0$. In the case of $n>0$, $b(0) = 0 \ \vee  \cdots \vee \ b(0) = n$.  Now we assume that for all $w\leq k$ the theorem holds, we proceed to show that the  theorem holds for $k+1$.  The following derivation will suffice:

\begin{prooftree}
\AxiomC{$\begin{array}{c}(IH) \\ c_{b}(k,n,y) \end{array}$}
\AxiomC{$\begin{array}{c}  (Cor. \ref{cor:second}[i\leftarrow b(k+1)]) \\  f(x_{k+1})= b(k+1), P \vdash  \end{array}$}
\RightLabel{$res(\sigma,P)$} 
\BinaryInfC{$\begin{array}{c}   c_{b}(k+1,n,z)  \end{array}$}
\end{prooftree}
$$P = f(m(k,\overline{x}_{k},max(s(x_{k+1}),t))) = b(k+1)$$
$$\sigma =\left\lbrace y  \leftarrow max(s(x_{k+1}),z)\right\rbrace $$
\\
$\square$

\subsection{SPASS Refutation of $C(n)$: Instance Four}
\label{sec:spassreffour}
The refutation provided in this section is almost identical to the output from SPASS except for a few minor changes to the syntax to aid reading. \\\\
$1[0:Inp] \ \|\|  \ \Rightarrow \ eq(f(U),3) \ , \ eq(f(U),2) \ , \ eq(f(U),1) \ , \ eq(f(U),0)*$\\\\
$2[0:Inp] \ \|\|  \ \Rightarrow \ le(U,U)*$\\\\
$3[0:Inp] \ \|\| \ le(max(U,V),W)* \ \Rightarrow \ le(U,W)$\\\\
$4[0:Inp] \ \|\| \ le(max(U,V),W)* \ \Rightarrow \ le(V,W)$\\\\
$5[0:Inp] \ \|\| \ le(s(U),V)*+ \ , \ eq(f(U),0)* \ , \ eq(f(V),0)* \ \Rightarrow $\\\\
$6[0:Inp] \ \|\| \ le(s(U),V)*+ \ , \ eq(f(U),1)* \ , \ eq(f(V),1)* \ \Rightarrow $\\\\
$7[0:Inp] \ \|\| \ le(s(U),V)*+ \ , \ eq(f(U),2)* \ , \ eq(f(V),2)* \ \Rightarrow $\\\\
$8[0:Inp] \ \|\| \ le(s(U),V)*+ \ , \ eq(f(U),3)* \ , \ eq(f(V),3)* \ \Rightarrow $\\\\
$9[0:Res:2.0,4.0] \ \|\|  \ \Rightarrow \ le(U,max(V,U))$\\\\
$10[0:Res:9.0,4.0] \ \|\|  \ \Rightarrow \ le(U,max(V,max(W,U)))$\\\\
$12[0:Res:2.0,3.0] \ \|\|  \ \Rightarrow \ le(U,max(U,V))$\\\\
$13[0:Res:9.0,3.0] \ \|\|  \ \Rightarrow \ le(U,max(V,max(U,W)))$\\\\
$15[0:Res:12.0,3.0] \ \|\|  \ \Rightarrow \ le(U,max(max(U,V),W))$\\\\
$16[0:Res:12.0,4.0] \ \|\|  \ \Rightarrow \ le(U,max(max(V,U),W))$\\\\
$19[0:Res:15.0,3.0] \ \|\|  \ \Rightarrow \ le(U,max(max(max(U,V),W),X))$\\\\
$20[0:Res:15.0,4.0] \ \|\|  \ \Rightarrow \ le(U,max(max(max(V,U),W),X))$\\\\
$23[0:Res:2.0,8.0] \ \|\| \ eq(f(U),3) \ , \ eq(f(s(U)),3)* \ \Rightarrow$\\\\
$25[0:Res:10.0,8.0] \ \|\| \ eq(f(U),3) \ , \ eq(f(max(V,max(W,s(U)))),3)* \ \Rightarrow$\\\\
$27[0:Res:12.0,8.0] \ \|\| \ eq(f(U),3) \ , \ eq(f(max(s(U),V)),3)* \ \Rightarrow$\\\\
$28[0:Res:15.0,8.0] \ \|\| \ eq(f(U),3) \ , \ eq(f(max(max(s(U),V),W)),3)* \ \Rightarrow$\\\\
$42[0:Res:2.0,7.0] \ \|\| \ eq(f(U),2) \ , \ eq(f(s(U)),2)* \ \Rightarrow$\\\\
$43[0:Res:9.0,7.0] \ \|\| \ eq(f(U),2) \ , \ eq(f(max(V,s(U))),2)* \ \Rightarrow$\\\\
$44[0:Res:10.0,7.0] \ \|\| \ eq(f(U),2) \ , \ eq(f(max(V,max(W,s(U)))),2)* \ \Rightarrow$\\\\
$50[0:Res:12.0,7.0] \ \|\| \ eq(f(U),2) \ , \ eq(f(max(s(U),V)),2)* \ \Rightarrow$\\\\
$52[0:Res:16.0,7.0] \ \|\| \ eq(f(U),2) \ , \ eq(f(max(max(V,s(U)),W)),2)* \ \Rightarrow$\\\\
$54[0:Res:19.0,4.0] \ \|\|  \ \Rightarrow \ le(U,max(max(max(max(V,U),W),X),Y))$\\\\
$59[0:Res:20.0,7.0] \ \|\| \ eq(f(U),2) \ , \ eq(f(max(max(max(V,s(U)),W),X)),2)* \ \Rightarrow$\\\\
$69[0:Res:2.0,6.0] \ \|\| \ eq(f(U),1) \ , \ eq(f(s(U)),1)* \ \Rightarrow$\\\\
$70[0:Res:9.0,6.0] \ \|\| \ eq(f(U),1) \ , \ eq(f(max(V,s(U))),1)* \ \Rightarrow$\\\\
$74[0:Res:13.0,6.0] \ \|\| \ eq(f(U),1) \ , \ eq(f(max(V,max(s(U),W))),1)* \ \Rightarrow$\\\\
$79[0:Res:16.0,6.0] \ \|\| \ eq(f(U),1) \ , \ eq(f(max(max(V,s(U)),W)),1)* \ \Rightarrow$\\\\
$89[0:Res:2.0,5.0] \ \|\| \ eq(f(U),0) \ , \ eq(f(s(U)),0)* \ \Rightarrow$\\\\
$90[0:Res:9.0,5.0] \ \|\| \ eq(f(U),0) \ , \ eq(f(max(V,s(U))),0)* \ \Rightarrow$\\\\
$98[0:Res:12.0,5.0] \ \|\| \ eq(f(U),0) \ , \ eq(f(max(s(U),V)),0)* \ \Rightarrow$\\\\
$123[0:Res:1.3,89.1] \ \|\| \ eq(f(U),0) \ \Rightarrow \ eq(f(s(U)),3) \ , \ eq(f(s(U)),2) \ , \ eq(f(s(U)),1)$\\\\
$159[0:Res:54.0,8.0] \ \|\| \ eq(f(U),3) \ , \ eq(f(max(max(max(max(V,s(U)),W),X),Y)),3)* \ \Rightarrow$\\\\
$196[0:Res:1.3,90.1] \ \|\| \ eq(f(U),0) \ \Rightarrow \ eq(f(max(V,s(U))),3) \\\\ eq(f(max(V,s(U))),2) \ , \ eq(f(max(V,s(U))),1)$\\\\
$197[0:Res:1.3,98.1] \ \|\| \ eq(f(U),0) \ \Rightarrow \ eq(f(max(s(U),V)),3) \\\\ eq(f(max(s(U),V)),2) \ , \ eq(f(max(s(U),V)),1)$\\\\
$423[0:Res:196.3,79.1] \ \|\| \ eq(f(U),0) \ , \ eq(f(V),1) \ \Rightarrow \\\\ eq(f(max(max(W,s(V)),s(U))),3) \ , \ eq(f(max(max(W,s(V)),s(U))),2)$\\\\
$450[0:Res:197.3,74.1] \ \|\| \ eq(f(U),0) \ , \ eq(f(V),1) \ \Rightarrow \\\\ eq(f(max(s(U),max(s(V),W))),3) \ , \ eq(f(max(s(U),max(s(V),W))),2)$\\\\
$955[0:Res:423.3,59.1] \ \|\| \ eq(f(U),0) \ , \ eq(f(V),1) \ , \ eq(f(W),2) \\\\ \Rightarrow \ eq(f(max(max(max(X,s(W)),s(V)),s(U))),3)$\\\\
$1009[0:Res:450.3,44.1] \ \|\| \ eq(f(U),0) \ , \ eq(f(V),1) \ , \ eq(f(W),2) \\\\ \Rightarrow \ eq(f(max(s(U),max(s(V),s(W)))),3)$\\\\
$2272[0:Res:955.3,159.1] \ \|\| \ eq(f(U),0)* \ , \ eq(f(V),1)* \ , \ eq(f(W),2)* \\\\ eq(f(X),3)* \ \Rightarrow$\\\\
$2273[0:MRR:1009.3,2272.3] \ \|\| \ eq(f(U),0)*+ \ , \ eq(f(V),1)* \\\\ eq(f(W),2)* \ \Rightarrow$\\\\
$2301[0:MRR:450.3,2273.2] \ \|\| \ eq(f(U),0) \ , \ eq(f(V),1) \\\\ \Rightarrow \ eq(f(max(s(U),max(s(V),W))),3)$\\\\
$2450[0:Res:2301.2,25.1] \ \|\| \ eq(f(U),0)* \ , \ eq(f(V),1)* \\\\ eq(f(W),3)* \ \Rightarrow$\\\\
$2459[0:MRR:2301.2,2450.2] \ \|\| \ eq(f(U),0)*+ \ , \ eq(f(V),1)* \ \Rightarrow$\\\\
$2577[0:MRR:123.3,2459.1] \ \|\| \ eq(f(U),0) \ \Rightarrow \ eq(f(s(U)),3) \\\\ eq(f(s(U)),2)$\\\\
$2578[0:MRR:196.3,2459.1] \ \|\| \ eq(f(U),0) \ \Rightarrow \\ eq(f(max(V,s(U))),3) \ , \ eq(f(max(V,s(U))),2)$\\\\
$2613[0:Res:2578.2,50.1] \ \|\| \ eq(f(U),0) \ , \ eq(f(V),2) \ \Rightarrow \ eq(f(max(s(V),s(U))),3)$\\\\
$2615[0:Res:2578.2,52.1] \ \|\| \ eq(f(U),0) \ , \ eq(f(V),2) \\ \Rightarrow \ eq(f(max(max(W,s(V)),s(U))),3)$\\\\
$2676[0:Res:2615.2,28.1] \ \|\| \ eq(f(U),0)* \ , \ eq(f(V),2)* \ , \ eq(f(W),3)* \ \Rightarrow$\\\\
$2684[0:MRR:2613.2,2676.2] \ \|\| \ eq(f(U),0)*+ \ , \ eq(f(V),2)* \ \Rightarrow$\\\\
$2714[0:MRR:2577.2,2684.1] \ \|\| \ eq(f(U),0) \ \Rightarrow \ eq(f(s(U)),3)$\\\\
$2715[0:MRR:2578.2,2684.1] \ \|\| \ eq(f(U),0) \ \Rightarrow \ eq(f(max(V,s(U))),3)$\\\\
$2749[0:Res:2715.1,27.1] \ \|\| \ eq(f(U),0)* \ , \ eq(f(V),3)* \ \Rightarrow$\\\\
$2764[0:MRR:2714.1,2749.1] \ \|\| \ eq(f(U),0)* \ \Rightarrow $\\\\
$2795[0:MRR:1.3,2764.0] \ \|\|  \ \Rightarrow \ eq(f(U),3) \ , \ eq(f(U),2) \ , \ eq(f(U),1)$\\\\
$2796[0:Res:2795.2,69.1] \ \|\| \ eq(f(U),1) \ \Rightarrow \ eq(f(s(U)),3) \ , \ eq(f(s(U)),2)$\\\\
$2797[0:Res:2795.2,70.1] \ \|\| \ eq(f(U),1) \ \Rightarrow \ eq(f(max(V,s(U))),3) \\ eq(f(max(V,s(U))),2)$\\\\
$2831[0:Res:2797.2,50.1] \ \|\| \ eq(f(U),1) \ , \ eq(f(V),2) \\ \Rightarrow \ eq(f(max(s(V),s(U))),3)$\\\\
$2833[0:Res:2797.2,52.1] \ \|\| \ eq(f(U),1) \ , \ eq(f(V),2) \\ \Rightarrow \ eq(f(max(max(W,s(V)),s(U))),3)$\\\\
$2896[0:Res:2833.2,28.1] \ \|\| \ eq(f(U),1)* \ , \ eq(f(V),2)* \ , \ eq(f(W),3)* \ \Rightarrow$\\\\
$2904[0:MRR:2831.2,2896.2] \ \|\| \ eq(f(U),1)*+ \ , \ eq(f(V),2)* \ \Rightarrow$\\\\
$2934[0:MRR:2796.2,2904.1] \ \|\| \ eq(f(U),1) \ \Rightarrow \ eq(f(s(U)),3)$\\\\
$2935[0:MRR:2797.2,2904.1] \ \|\| \ eq(f(U),1) \ \Rightarrow \ eq(f(max(V,s(U))),3)$\\\\
$2969[0:Res:2935.1,27.1] \ \|\| \ eq(f(U),1)* \ , \ eq(f(V),3)* \ \Rightarrow$\\\\
$2984[0:MRR:2934.1,2969.1] \ \|\| \ eq(f(U),1)* \ \Rightarrow$\\\\
$3015[0:MRR:2795.2,2984.0] \ \|\|  \ \Rightarrow \ eq(f(U),3) \ , \ eq(f(U),2)$\\\\
$3016[0:Res:3015.1,42.1] \ \|\| \ eq(f(U),2) \ \Rightarrow \ eq(f(s(U)),3)$\\\\
$3017[0:Res:3015.1,43.1] \ \|\| \ eq(f(U),2) \ \Rightarrow \ eq(f(max(V,s(U))),3)$\\\\
$3050[0:Res:3017.1,27.1] \ \|\| \ eq(f(U),2)* \ , \ eq(f(V),3)* \ \Rightarrow$\\\\
$3065[0:MRR:3016.1,3050.1] \ \|\| \ eq(f(U),2)* \ \Rightarrow$\\\\
$3096[0:MRR:3015.1,3065.0] \ \|\|  \ \Rightarrow \ eq(f(U),3)$\\\\
$3098[0:MRR:23.1,23.0,3096.0] \ \|\|  \ \Rightarrow$\\\\

\subsection{Growth Rate of Refutation}
\label{growthrate}
\begin{definition}
\index{$Occ(\cdot ,\cdot)$}
Let $Occ(x,r)$ be defined as the number of times the clause $x$ is used in the refutation $r$.
\end{definition}

\begin{theorem}
Let $r$ be the resolution refutation of Thm. \ref{thm:refofC} for the clause set $C(n)$, then $Occ(C5,r)$ is the result of the following recurrence relation $a(n+1) = (n+1)*a(n) + 1$ and $a(0)=1$.
\end{theorem}
\begin{proof}
Let us consider the case for the clause set $C(0)$. This is the case when we have only one symbol in the function's range. If we compute the recurrence we get $a(1) = a(0) +1 = 2$ Now let us assume it holds for all $m\leq n$ and show it hold for $n+1$. In the proof of Lem. \ref{lem:l13}, when  deriving $c'_{b}(0,0)$ the literal $f(\alpha) = b(0)$ is in the antecedent for every clause higher in the resolution derivation and it is  never used in a resolution step . If we remove this clause from the antecedent then we have a resolution refutation for the clause $C(n)$, only if we rename the schematic sort terms accordingly. To refute $C(n+1)$ we need to derive $n+1$ distinct $c'_{b}(0,0)$ clauses and resolve them with a single instance of $(C5)$. Thus, we have the equation, $Occ(C5^{n+1},r_{n+1}) = (n+1)*Occ(C5^{n},r_{n}) +1 $   where $r_{n+1}$ is the  resolution refutation of Thm. \ref{thm:refofC} for the clause set $C(n+1)$ and  $r_{n}$ is the resolution refutation of Thm. \ref{thm:refofC} for the clause set $C(n)$. Thus, the theorem holds by induction. 
\\$\square$ 
\end{proof}

\begin{corollary}
\index{Recurrence Relation}
The recurrence relation  $a(n) = n\cdot a(n-1) + 1$ and $a(0)=1$ is equivalent to the equation:

$$f(n) =n!\cdot \sum_{i=0}^{n} \frac{1}{i!}$$
\end{corollary}
\begin{proof}
If we unroll the relation one we get,
$$ a(n) =  n\cdot (n-1)\cdot a(n-2) + n + 1  =  n\cdot (n-1)\cdot a(n-2) + \frac{n!}{(n-1)!} + \frac{n!}{n!}$$

Thus, unrolling the function $k$ times results in the following:

$$ a(n) = \left(  \prod^{n}_{i=n-k+1} i\right) \cdot a(n-k) + \sum^{n}_{i=n-k+1} \frac{n!}{i!} $$

Now when we set $k=n$ we get, 

$$ a(n) = n! + \sum^{n}_{i=1} \frac{n!}{i!} = \frac{n!}{0!}+ \sum^{n}_{i=1} \frac{n!}{i!}  = \sum^{n}_{i=0} \frac{n!}{i!}  $$ \\ $\square$
\end{proof}

\end{document}